\documentclass[11pt,a4paper,twoside,reqno]{amsart}
\makeatletter
\def\@settitle{\begin{center}%
		\baselineskip14\p@\relax
		\normalfont\LARGE\scshape\bfseries
		\@title
	\end{center}%
}

\def\@setauthors{%
  \begingroup
  \def\thanks{\protect\thanks@warning}%
  \trivlist
  \centering\footnotesize \@topsep30\p@\relax
  \advance\@topsep by -\baselineskip
  \item\relax
  \author@andify\authors
  \def\\{\protect\linebreak}%
  \authors%
  \ifx\@empty\contribs
  \else
    ,\penalty-3 \space \@setcontribs
    \@closetoccontribs
  \fi
  \endtrivlist
  \endgroup
}

\makeatother
\makeatletter

\def\subsection{\@startsection{subsection}{2}%
	\z@{.5\linespacing\@plus.7\linespacing}{.5\linespacing}%
	{\normalfont\large\bfseries}}

\def\subsubsection{\@startsection{subsubsection}{3}%
	\z@{.5\linespacing\@plus.7\linespacing}{.5\linespacing}%
	{\normalfont\itshape}}

\usepackage[usenames, dvipsnames]{color}
\definecolor{darkblue}{rgb}{0.0, 0.0, 0.45}

\usepackage[colorlinks	= true,
raiselinks	= true,
linkcolor	= darkblue, 
citecolor	= Mahogany,
urlcolor	= ForestGreen,
pdfauthor	= {Peyman Mohajerin Esfahani},
pdftitle	= {},
pdfkeywords	= {},
pdfsubject	= {},
plainpages	= false]{hyperref}

\usepackage{dsfont,amssymb,amsmath,subfigure, graphicx,enumitem,multirow} 
\usepackage{amsfonts,dsfont,mathtools, mathrsfs,amsthm} 
\usepackage[amssymb, thickqspace]{SIunits}
\usepackage{algorithm,fancyhdr}
\usepackage{algorithmicx}
\usepackage{algpseudocode}

\allowdisplaybreaks
\date{\today}

\theoremstyle{theorem}
\newtheorem{Thm}{Theorem}[section]
\newtheorem{Prop}[Thm]{Proposition}

\newtheorem{Cor}[Thm]{Corollary}

\newtheorem{Rem}[Thm]{Remark}

\newcommand{\mPr}{\mathbf{Pr}}
\newcommand{\mE}{\mathbf{E}}

\newcommand{\mfq}{\mathfrak{q}}

\newcommand{\F}{\mathds{F}}
\newcommand{\T}{\mathds{T}}

\newcommand{\R}{\mathbb{R}}

\newcommand{\N}{\mathbb{N}}

\usepackage{epsfig} 
\usepackage{graphicx}
\usepackage{grffile}
\usepackage{amsmath}
\usepackage{enumitem}
\usepackage{booktabs}
\usepackage{multirow}
\usepackage{array}
\usepackage{cite}
\usepackage{bm}
\usepackage{pdfpages} 
\usepackage{subfigure}
\usepackage{booktabs}
\usepackage{xcolor}
\usepackage{algorithm}
\usepackage{algpseudocode}
\usepackage{wrapfig}
\usepackage{caption}
\usepackage{appendix}

\begin{document}

\title[]{Real-Time Ground Fault Detection \\ for Inverter-Based Microgrid Systems}
\author{Jingwei Dong, Yucheng Liao, Haiwei Xie, Jochen Cremer, and Peyman Mohajerin Esfahani}
\thanks{The authors are with the Delft Center for Systems and Control (\{J.Dong-6, P.MohajerinEsfahani\}@tudelft.nl, Y.Liao-3@student.tudelft.nl) and the DAI Energy Lab (\{H.Xie-2, J.L.Cremer\}@tudelft.nl), Delft University of Technology, The Netherlands. This work is supported by the ERC (European Research Council) grant TRUST-949796 and CSC (China Scholarship Council) with funding number: 201806120015. The authors are grateful to M. Popov and A. Lekic for the valuable discussion on inverter-based microgrid systems.}
\maketitle

\begin{abstract}
    Ground fault detection in inverter-based microgrid~(IBM) systems is challenging, particularly in a real-time setting, as the fault current deviates slightly from the nominal value. 
    This difficulty is reinforced when there are partially decoupled disturbances and modeling uncertainties.
    The conventional solution of installing more relays to obtain additional measurements is costly and also increases the complexity of the system.  
    In this paper, we propose a data-assisted diagnosis scheme based on an optimization-based fault detection filter with the output current as the only measurement.
    Modeling the microgrid dynamics and the diagnosis filter, we formulate the filter design as a quadratic programming (QP) problem that accounts for decoupling partial disturbances, robustness to non-decoupled disturbances and modeling uncertainties by training with data, and ensuring fault sensitivity simultaneously. To ease the computational effort, we also provide an approximate but analytical solution to this QP. 
    Additionally, we use classical statistical results to provide a thresholding mechanism that enjoys probabilistic false-alarm guarantees.
    Finally, we implement the IBM system with Simulink and Real Time Digital Simulator~(RTDS) to verify the effectiveness of the proposed method through simulations.
\end{abstract}

\section{Introduction}
In the past decade, IBM systems have gained popularity as power systems become increasingly complex and rely more on renewable energy sources~\cite{altaf2022microgrid}.
These microgrid systems help integrate renewable energy sources into power systems and regulate the amount of power supplied to customers to provide high-quality power and reduce energy costs. 
They can also operate independently and allow for local control of distributed generation, for example, when the main grid is unavailable due to blackouts or storms~\cite{zamani2012communication}.
This greatly increases the reliability of power systems.

Although IBM systems offer many benefits, they are susceptible to faults that can pose safety risks and damage equipment. 
However, the conventional protection strategy for power systems, such as overcurrent detection, is inefficient in detecting faults in IBM systems~\cite{lai2015comprehensive}.
This is because the fault current only slightly deviates from the nominal value due to a fault current limiter embedded in the inverter controller~\cite{karimi2019protection}.
The fault detection problem is more difficult when considering disturbances that cannot be completely decoupled and modeling uncertainties.
Therefore, developing an effective fault detection scheme for IBM systems in the presence of partially decoupled disturbances and modeling uncertainties remains a challenge, particularly when the output current is the only measurement.
In this paper, we focus on the detection of ground faults as they are the most common and problematic type of faults in IBM systems~\cite{loix2009protection}.

To address the fault detection problem for microgrid systems, researchers have developed several differential methods that rely on communication infrastructure between relays. 
These methods measure differences in the current symmetrical components~\cite{casagrande2013differential}, the energy content of current~\cite{samantaray2012differential}, the instantaneous current with comparative voltage~\cite{sortomme2009microgrid}, and the traveling wave polarities~\cite{liu2020fast} to detect faults.
Though these methods have shown effectiveness, relying on communication devices can reduce the reliability of systems, and it can be expensive and time-consuming to implement and maintain new equipment.
In addition to differential methods, active fault detection methods have emerged as another popular solution to fault detection for microgrid systems in recent years.
By introducing carefully designed input signals into the system, active fault detection methods can enhance the detectability of faults. 
In~\cite{karimi2008negative}, the authors inject a small negative-sequence current ($<3\%$) into the microgrid system and detect faults by using a signal processing technique to quantify the resulting negative-sequence voltage.
Most recently, to detect ground faults, the authors in~\cite{pirani2022optimal} provided an optimal input design method ensuring that the output sets of normal and faulty modes of an IBM system are separated with probabilistic guarantees. 
However, the injected input can degrade system performance, and it is unsuitable for online monitoring due to the intensive computation required for generating the input sequence.

In contrast to differential methods, fault detection methods based on residual generation are less dependent on the communication infrastructure and additional sensors. 
Moreover, residual generation-based methods are more suitable for online monitoring than active input design because they do not require continuous updates and have no impact on system performance.
In the fault detection field, residual generators are generally constructed using observer-based or \mbox{parity-space} methods~\cite{gao2015survey}. 
Researchers employ optimization techniques to determine the parameters of residual generators, such that the residuals are sensitive to faults while being robust against disturbances and uncertainties. Alternatively, decomposition techniques such as unknown input observers~(UIO)~\cite{gao2015unknown} can be utilized to decouple disturbances from residuals.
However, we found that the UIO approach could fail to satisfy the detectability condition when applied to IBM systems with a limited number of measurements. 

In~\cite{nyberg2006residual}, the authors proposed a parity-space-like approach to designing residual generators in the framework of linear differential-algebraic equations (DAE).
The derived residual generators can have a lower order than that of the system, thus reducing the computational complexity when dealing with large-scale systems. 
In addition, this framework provides design freedom in the sense that one can transform the design of residual generators into different optimization problems to obtain desired solutions based on specific requirements. For example, the authors in~\cite{esfahani2015tractable}
reformulated the robust fault detection filter design for nonlinear systems as a QP problem to decouple disturbances and minimize the effects of nonlinearity on residuals.
Based on this, results for attack detection~\cite{pan2021dynamic}, diagnosis of switched systems~\cite{dong2023multimode}, and multiple fault estimation~\cite{van2022multiple} have been developed in the DAE framework as well.  
We would like to emphasize that these methods~\cite{esfahani2015tractable,dong2023multimode,pan2021dynamic,van2022multiple} rely on an accurate system model and all consider disturbances that can be completely decoupled. 
However, in reality, modeling uncertainties are unavoidable and disturbances generally cannot be completely decoupled, all of which pose challenges to fault diagnosis tasks.

\textbf{Main contributions:} 
In this work, we take advantage of the DAE framework to design filters for ground fault detection in the IBM system. 
To the best of our knowledge, this is the first attempt to design fault detection filters within the DAE framework that enables real-time monitoring of ground faults in the IBM system with partially decoupled disturbances and modeling uncertainties. 
The contributions of this paper are summarized as follows: 
\begin{itemize}
\item {\bf Dynamic system modeling:} 
    We develop a unified \mbox{state-space} model for the IBM system in both normal mode and the presence of ground faults (Sections~\ref{subsec: fault-free model},~\ref{subsec: faulty model}). This model is further formulated in the DAE framework, which facilitates the design of robust fault detection filters.   

\item {\bf Data-assisted disturbance and uncertainty rejection:}
    To address partially decoupled disturbances and uncertainties, we borrow the idea from \cite[Approach (II)]{esfahani2015tractable} to reframe filter design as a QP problem. The reformulation enables us to decouple partial disturbances, mitigate the effects of non-decoupled disturbances and modeling uncertainties by training with data, and ensure fault sensitivity~(Theorem~\ref{thm: FD QP}).  
    Inspiring from \cite[Corollary 1]{van2021real}, we also obtain an approximate analytical solution to this QP problem with arbitrary accuracy (Corollary \ref{cor: appro sol}), allowing for online updates of filter parameters.
    
    \item {\bf Probabilistic false alarm certificate:}   
    Leveraging the classical Markov inequality, we further propose a threshold determination method along with probabilistic \mbox{false-alarm} guarantees (Proposition~\ref{prop:prob perform}).
    
    \item {\bf Validation through a high-fidelity simulator:} 
    To validate the effectiveness of the proposed diagnosis scheme, we test it on an IBM system constructed using Simulink and RTDS, which can effectively simulate the practical characteristics of smart grids.
\end{itemize}

The rest of the paper is organized as follows. 
The modeling of an IBM system and the problem formulation are presented in Section~\ref{sec: modeling}. 
In Section~\ref{sec:main results}, we provide the design method for the fault detection filter. 
In Section~\ref{sec:simulation}, we evaluate the effectiveness of the proposed approaches with simulations.
Finally, Section~\ref{sec:conslusion} concludes the paper with future directions. 

\textbf{Notation:} 
Sets~$\R (\R_+)$ and $\N$ denote all (positive) reals and non-negative integers, respectively. 
The space of~$n$ dimensional real-valued vectors is denoted by~$\R^n$.
For a vector~$v = [v_1,\dots,v_n]$, the infinity-norm of $v$ is~$\|v\|_\infty = \max_{i\in\{1,\dots,n\}}|v_i|$.
The diagonal operator is denoted by~$\textup{diag}(\cdot)$.
For two discrete-time signals~$s_1$ and~$s_2$ taking values in $\mathbb{R}^{n}$ with length~$T$, the $\mathcal{L}_{2}$ inner product is represented as $\langle s_1, s_2 \rangle:=\sum_{k=1}^{T} s_1^{\top}(k) s_2(k)$, and the corresponding norm~$\| s_1 \|_{\mathcal{L}_{2}} := \sqrt{\langle s_1, s_1 \rangle}$.
The notation~$\mathbf{0}_{m\times n}$ denotes a zero matrix with $m$ rows and $n$ columns. 
The identity matrix with an appropriate dimension is denoted by~$I$.
For a random variable~$\chi$, $\mPr[\chi]$ and $\mE[\chi]$ are the probability law and the expectation of $\chi$.

\section{Modeling and Problem Statement}\label{sec: modeling}
In this section, we present the state-space model of an IBM system and consider three-phase symmetrical ground faults. Then, we formulate the problem addressed in this work.

\subsection{System description}
An IBM generally consists of four components: the power supplier, the LCL filter, the controller, and the load, as shown in Figure~\ref{fig:Architecture}. 
Let us elaborate on the functions of each component. 
\begin{enumerate}
    \item \textbf {Power supplier}: 
    The power supply part provides power to the microgrid by following a reference voltage $v^*_i$ from the current controller. 
    It contains a distributed generator (DG) source and an inverter.
    In this work, we assume that (1) an ideal DG source is available, and (2) the inverter switching process can be neglected due to its high switching frequency. 
    Therefore, instead of modeling the generator and inverter, we can set the output voltage of the inverter directly to $v_i = v^*_i$.
    The real-time output current of the inverter is denoted by $i_l$. 
    As the single DG source supplies all power to the load, droop control is unnecessary, and the microgrid frequency $\omega$ is constant. This differs from \cite{pogaku2007modeling}, where multiple DG sources operate simultaneously.

    \item \textbf{LCL filter}: 
    The LCL filter is used to filter the harmonics produced by the inverter. It consists of two resistors $R_f$ and $R_c$, two inductors $L_f$ and $L_c$, and a capacitor $C_f$. 
    The signals $v_o$ and $i_o$ denote the grid-side voltage and the output current, respectively.
    
    \item \textbf{Controller}: 
    The control part keeps the grid-side voltage at some reference voltage~$v^*_o$, where~$v^*_o$ is determined by load demand and generation capacity. 
    This can be achieved through an inner current controller and an outer voltage controller, which are all PI controllers~\cite{leitner2017small}. 
    The outer voltage controller sets reference $i^*_l$ for the inner current controller. 
    The fault current limiter (FCL) is a saturation block that protects the microgrid from large fault currents. 
    
    \item \textbf {Load}: 
    The load denoted by~$R_{L}$ is purely resistive, and~$\Delta R_L$ is the unknown load change.
\end{enumerate}

\begin{figure}[t]
\centering
\includegraphics[scale=1.2]{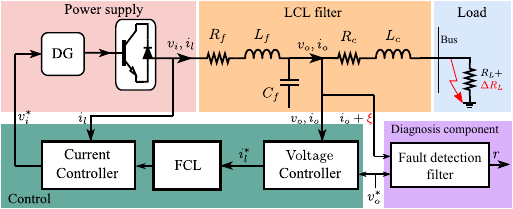}
\caption{Architecture of an IBM system with the diagnosis component.}
\label{fig:Architecture}
\end{figure}  

The mentioned voltages and currents are based on a three-phase system. 
We introduce the direct-quadrature ($dq$) transform to simplify the analysis.
Specifically, for a three-phase system with current~$i = \left[ i_a ~i_b ~i_c\right]^{\top}$ and voltage~$v = \left[ v_a ~v_b ~v_c\right]^{\top}$ in the $abc$ framework, the $dq$ transform projects~$i$ and~$v$ onto $dq$-axis, i.e., $i_{d q}=\mathbf{P}i, ~v_{d q}=\mathbf{P}v$, where~$i_{dq} = \left[ i_d ~i_q\right]^{\top}, ~v_{dq} = \left[ v_d ~v_q\right]^{\top}$. The projection matrix $\mathbf{P}$ is given by
\begin{equation*}
\mathbf{P}= \frac{2}{3}
\begin{bmatrix}
    \cos (\theta) & \cos \left(\theta-\frac{2 \pi}{3}\right) & \cos \left(\theta+\frac{2 \pi}{3}\right) \\
\sin (\theta) & \sin \left(\theta-\frac{2 \pi}{3}\right) & \sin \left(\theta+\frac{2 \pi}{3}\right)
\end{bmatrix},
\end{equation*}
in which $\theta$ is the constantly changing angle between the $d$-axis and the $a$-axis. 
We refer interested readers to~\cite{park1929two} for more details about the $dq$ transformation. 
For simplicity of expression, we add subscripts to indicate the variables after $dq$ transformation throughout the paper, e.g.,~$v \xrightarrow[]{dq} v_{dq}$, $i \xrightarrow[]{dq} i_{dq}$, and so forth.

\subsection{State-space model of the fault-free IBM system}\label{subsec: fault-free model}
To obtain the state-space model of the fault-free IBM system, we first model individual components of the microgrid including the voltage controller, the current controller, and the LCL filter in this subsection. 
Let us start with the voltage controller in the control component.
We transform~$v_o$,~$v_o^*$,~$i_o$ and~$i_l^*$ into the $dq$ framework, which are~$v_{odq},~v_{odq}^*,~i_{odq}$ and~$i_{ldq}^*$, respectively.
The cumulative error between~$v_{odq}$ and~$v^*_{odq}$ denoted by~$\phi_{dq} := [ \phi_d ~\phi_q]^{\top}$ can be written as
\begin{equation}
\label{eqa:phiDef}
\frac{\mathrm{d} \phi_{d}(t)}{\mathrm{d} t}=v_{o d}^{*}(t)-v_{o d}(t), \quad \frac{\mathrm{d} \phi_{q}(t)}{\mathrm{d} t}=v_{o q}^{*}(t)-v_{o q}(t).
\end{equation}
Considering that the voltage controller is a PI controller, we have the following relations
\begin{align}\label{eqa:voltageController}
\left\{ \begin{array}{l}
     i_{l d}^{*} (t) = F i_{o d}(t)-\omega C_{f} v_{o q}(t) + K^v_P \left(v_{o d}^{*}(t) - v_{o d}(t)\right)+K^v_I \phi_{d}(t), \\
     i_{l q}^{*}(t)= F i_{o q}(t)+\omega C_{f} v_{o d}(t)+K^v_P\left(v_{o q}^{*}(t)-v_{o q}(t)\right)+K^v_I \phi_{q}(t),
\end{array} \right.
\end{align}
where~$F$ is the feedforward coefficient,~$K^v_P$ and $K^v_I$ denote the proportional and integral gains of the voltage controller, respectively. From~\eqref{eqa:phiDef} and \eqref{eqa:voltageController}, we obtain the state-space model of the voltage controller
\begin{align}\label{eqa:ss4VoltageCtrl}
\left\{ \begin{array}{l}
     \dot{\phi}_{dq}(t) = B_{v 1} v_{o d q}^{*}(t) +B_{v 2} \begin{bmatrix}
         i_{l d q}(t) &v_{o d q}(t) &i_{o d q}(t)
     \end{bmatrix}^{\top},\\
    i_{l d q}^{*}(t)  = C_{v} \phi_{dq}(t) +D_{v 1} v_{o d q}^{*}(t) +D_{v 2}\begin{bmatrix}
        i_{l d q}(t)&  v_{o d q}(t)& i_{o d q}(t)
    \end{bmatrix}^{\top},
\end{array} \right.
\end{align}
where the matrices are 
\begin{align*}
    B_{v 1} &= \begin{bmatrix}
        1 & 0 \\
        0 & 1
    \end{bmatrix},
    ~B_{v 2} = \begin{bmatrix}
        0 & 0 & -1 & 0 & 0 & 0 \\
        0 & 0 & 0 & -1 & 0 & 0
    \end{bmatrix},
    ~C_v =\begin{bmatrix}
        K_I^v & 0 \\
        0 & K_I^v
        \end{bmatrix}, \\
    D_{v 1} &= \begin{bmatrix}
        K_P^v & 0 \\
        0 & K_P^v
    \end{bmatrix},
    ~D_{v 2} = \begin{bmatrix}
        0 & 0 & -K_P^v & -\omega C_{f} & F & 0 \\
        0 & 0 & \omega C_{f} & -K_P^v & 0 & F
    \end{bmatrix}.
\end{align*}

Similarly, one can obtain the state-space model of the current controller. 
Let us transform~$i_l$,~$i_l^*$ and~$v_{i}^*$ into the $dp$ framework, which are~$i_{ldq},~i^*_{ldq}$ and~$v_{idq}^*$, respectively. 
The cumulative error between~$i_{ldq}$ and~$i_{ldq}^*$ is denoted by~$\gamma_{dq} := [\gamma_d ~\gamma_q]^{\top}$, i.e.,
\begin{equation}
\label{eqa:gaDef}
\frac{\mathrm{d} \gamma_{d}(t)}{\mathrm{d} t} = i_{l d}^{*}(t) - i_{l d}(t), \quad \frac{\mathrm{d} \gamma_{q}(t)}{\mathrm{d} t} = i_{l q}^{*}(t)-i_{l q}(t).
\end{equation}
Then, the dynamics of the current controller follows 
\begin{align}\label{eqa:currentController}
\left\{ \begin{array}{l}
     v_{i d}^{*}(t) = -\omega L_{f} i_{l q}(t) + K^c_P\left(i_{l d}^{*}(t) - i_{l d}(t)\right)+K^c_I \gamma_{d}(t),\\
     v_{i q}^{*}(t) = \omega L_{f} i_{l d}(t) + K^c_P(i_{l q}^{*}(t) - i_{l q}(t))+K^c_I \gamma_{q}(t),
\end{array}\right.
\end{align}
where $K^c_P$ and $K^c_I$ denote the proportional and integral gains of the current controller, respectively. 
Based on~\eqref{eqa:gaDef} and~\eqref{eqa:currentController}, the state-space model of the current controller is given by
\begin{align}\label{eqa:ss4CurrentCtrl}
\left\{ \begin{array}{l}
     \dot{\gamma}_{d q}(t) = B_{c 1} i_{l d q}^{*}(t) + B_{c 2} \begin{bmatrix}
         i_{l d q}(t) & v_{o d q}(t) & i_{o d q}(t)
     \end{bmatrix}^{\top}, \\
     v_{i d q}^{*}(t) = C_{c} \gamma_{d q}(t) + D_{c 1} i_{l d q}^{*}(t)  +D_{c 2} \begin{bmatrix}
         i_{l d q}(t) & v_{o d q}(t) & i_{o d q}(t)
     \end{bmatrix}^{\top} ,
\end{array}\right.
\end{align}
where
\begin{align*}
    B_{c 1}&= \begin{bmatrix}
            1 & 0 \\
            0 & 1
            \end{bmatrix},
    ~B_{c 2} =\begin{bmatrix}
        -1 & 0 &\mathbf{0}_{1\times4} \\
        0 & -1 &\mathbf{0}_{1\times4}
        \end{bmatrix}, 
    ~C_c =\begin{bmatrix}
        K_I^c & 0 \\
        0 & K_I^c
        \end{bmatrix},\\
    D_{c 1}&=\begin{bmatrix}
        K_P^c & 0 \\
        0 & K_P^c
        \end{bmatrix}, 
    ~D_{c 2} =\begin{bmatrix}
        -K_P^c & -\omega L_{f}	&\mathbf{0}_{1\times4} \\
        \omega L_{f} & -K_P^c   &\mathbf{0}_{1\times4}
        \end{bmatrix}.
\end{align*}

For the LCL filter modeling, we transform the output voltage of the inverter~$v_i$ and the bus voltage~$v_b$ into the $dq$ framework, i.e.,~$v_{idq}$ and~$v_{bdq}$, respectively. The dynamics of the LCL filter follows
\begin{align*}
\left\{ \begin{array}{l}
    \dot{i}_{l d}(t) =\frac{-R_{f}}{L_{f}} i_{l d}(t)+\omega i_{l q}(t)+\frac{1}{L_{f}} v_{i d}(t)-\frac{1}{L_{f}} v_{o d}(t), \\
    \dot{i}_{l q}(t) =\frac{-R_{f}}{L_{f}} i_{l q}(t)-\omega i_{l d}(t)+\frac{1}{L_{f}} v_{i q}(t)-\frac{1}{L_{f}} v_{o q}(t), \\
    \dot{v}_{o d}(t) =\omega v_{o q}(t)+\frac{1}{C_{f}} i_{l d}(t)-\frac{1}{C_{f}} i_{o d}(t) , \\
    \dot{v}_{o q}(t) =-\omega v_{o d}(t)+\frac{1}{C_{f}} i_{l q}(t)-\frac{1}{C_{f}} i_{o q}(t) , \\
    \dot{i}_{o d}(t)=\frac{-R_{c}}{L_{c}} i_{o d}(t)+\omega i_{o q}(t)+\frac{1}{L_{c}} v_{o d}(t)-\frac{1}{L_{c}} v_{b d}(t), \\
    \dot{i}_{o q}(t)=\frac{-R_{c}}{L_{c}} i_{o q}(t)-\omega i_{o d}(t)+\frac{1}{L_{c}} v_{o q}(t)-\frac{1}{L_{c}} v_{b q}(t) .
\end{array}\right.
\end{align*}
Then, the state-space model of the LCL filter becomes
\begin{align} \label{eqa:ss4LCL}
\left[\begin{array}{c}
\dot{i}_{l d q}(t) \\
\dot{v}_{o d q}(t) \\
\dot{i}_{o d q}(t)
\end{array}\right] =A_l\left[\begin{array}{c}
i_{l d q}(t) \\
v_{o d q}(t) \\
i_{o d q}(t)
\end{array}\right]+
\left[\begin{array}{cc}
B_{l1} & B_{l2}
\end{array}\right]
\left[\begin{array}{c}
     v_{i d q}(t)   \\
     v_{b d q}(t)
\end{array}\right] \ ,
\end{align}
where the bus voltage~$v_{bdq} = \text{diag} (R_L+\Delta R_L,R_L+\Delta R_L)  i_{odq}$, and the matrices are
\begin{align*}
A_{l}= \begin{bmatrix}
-\frac{R_{f}}{L_{f}} & \omega & -\frac{1}{L_{f}} & 0 & 0 & 0 \\
-\omega & -\frac{R_f}{L_f} & 0 & -\frac{1}{L_{f}} & 0 & 0 \\
\frac{1}{C_{f}} & 0 & 0 & \omega & -\frac{1}{C_{f}} & 0 \\
0 & \frac{1}{C_{f}} & -\omega & 0 & 0 & -\frac{1}{C_{f}} \\
0 & 0 & \frac{1}{L_{c}} & 0 & -\frac{R_{c}}{L_{c}} & \omega \\
0 & 0 & 0 & \frac{1}{L_{c}} & -\omega & -\frac{R_{c}}{L_{c}}
\end{bmatrix},  
~B_{l 1}=\begin{bmatrix}
\frac{1}{L_{f}} & 0 \\
0 &\frac{1}{L_{f}}\\
\mathbf{0}_{4\times 1} &\mathbf{0}_{4\times 1}  
\end{bmatrix},
~B_{l 2}=\begin{bmatrix}
\mathbf{0}_{4 \times 1} &\mathbf{0}_{4 \times 1}\\
-\frac{1}{L_{c}} & 0 \\
0 & -\frac{1}{L_{c}}
\end{bmatrix}.
\end{align*}

Recall that~$v_i = v^*_i$, and thus~$v_{idq} = v_{idq}^*$. 
By combining the derived models~\eqref{eqa:ss4VoltageCtrl},~\eqref{eqa:ss4CurrentCtrl}, and~\eqref{eqa:ss4LCL}, we obtain the complete state-space model of the inverter-based microgrid system in the fault-free case
\begin{align}\label{eqa:ss4inverter}
   \left\{ \begin{array}{l}
        \dot{x} (t)  =A_h x (t) + B_h v_{odq}^{*}(t) + B_d d(t), \\
        i_{odq} (t) =C x(t),
   \end{array}\right.
\end{align}
where~$x(t) = \begin{bmatrix}\phi^{\top}_{dq}(t) &\gamma^{\top}_{dq}(t) &i^{\top}_{ldq}(t)  & v^{\top}_{odq}(t) & i^{\top}_{odq}(t) \end{bmatrix}^{\top}$ is the augmented state of the microgrid system and~$d(t) = i_{odq}(t) \Delta R_L$ denotes the disturbance. The system matrices~$A_h,~B_h$, and~$C$ are given by
\begin{align*}
    &A_h=\begin{bmatrix}
    \mathbf{0}_{2\times 2} & \mathbf{0}_{2\times2} & B_{v 2} \\
    B_{c 1} C_{v} & \mathbf{0}_{2\times2} & B_{c 1} D_{v 2}+B_{c 2} \\
    B_{l1} D_{c 1} C_{v} & B_{l1} C_{c} &A_{h33}
    \end{bmatrix},
    ~B_h = \begin{bmatrix}
    B_{v1} \\ B_{c1}D_{v1} \\ B_{l1} D_{c1} D_{v1}
    \end{bmatrix},
    ~B_d = \begin{bmatrix}
        \mathbf{0}_{8 \times 1} &\mathbf{0}_{8 \times 1}\\ -\frac{1}{L_c} &0 \\0 &-\frac{1}{L_c}
    \end{bmatrix},\\
    &C= \begin{bmatrix}
        \mathbf{0}_{2 \times 8} &I
    \end{bmatrix},
\end{align*}
where~$A_{h33} = A_{l}+ B_{l1}\left(D_{c 1} D_{v 2}+D_{c 2}\right)+
B_{l2} \begin{bmatrix} R_L &0\\ 0 &R_L \end{bmatrix} 
\begin{bmatrix}
    \mathbf{0}_{2 \times 4} & I
\end{bmatrix}$. 

We would like to highlight that the number of states is~$10$, while we only have $2$ measurements, i.e.,~$i_{od}$ and $i_{oq}$.
We take into account disturbances resulting from unknown load changes, which are commonly observed in microgrid systems and typically manifest as random step signals~\cite{kahrobaeian2012interactive,singh2017distributed,xu2017robust}.
Therefore, in our study, we assume that~$d$ is a step signal taking random values uniformly within the known range~$[d_{lb},d_{ub}]$  where~$d_{lb},d_{ub} \in \mathbb{R}^2$.
Additionally, since the dimension of the measurement signal is equal to that of the disturbance,~$d$ cannot be fully decoupled~\cite[Chapter 6]{ding2008model}, leading to challenges in fault detection.  
To address this issue, we split~$B_d = [\hat{B}_d ~\check{B}_d]$ and define~$d = [\hat{d} ~\check{d}]^{\top}$, where~$\hat{d}$ and~$\check{d}$ represent the decoupled and non-decoupled parts, respectively. 

\begin{Rem}[Disturbance decoupling condition]\label{rem: decouple condition}
    Let~$\T_{d i_{odq}}$ denote the transfer function from the disturbance~$d$ to the measurement~$i_{odq}$, and~$\text{Rank}(\T_{d i_{odq}})$ denotes the rank of~$\T_{d i_{odq}}$. 
    According to \cite[Chapter 6]{ding2008model},~$d$ can be decoupled from~$i_{odq}$ if the number of unknown inputs is smaller than the number of sensors, i.e.,~$\text{Rank}(\T_{d i_{odq}}) < 2$. 
    Therefore,~$d$ can be decoupled from~$i_{odq}(t)$ when~$d(t)$ is a one-dimensional signal but not for a two-dimensional (and higher-dimensional) disturbance.
\end{Rem}

\subsection{State-space model of the IBM system with ground faults}\label{subsec: faulty model}
We consider three-phase symmetrical ground faults which can cause a short circuit and a sharp increase in currents. 
We know that after ground faults occur: (1) the bus voltage~$v_{bdq}=0$ because of the short circuit; and (2) the output of the voltage controller~$i_{ldq}^*$ saturates to a constant value~$\tau_{dq}$ immediately, i.e.,~$i_{ldq}^*(t) = \tau_{dq}$ for~$t \geq t_f$, where $t_f$ denotes the time instant when ground faults occur. 
Therefore, the state-space model of the current controller~\eqref{eqa:ss4CurrentCtrl} in the faulty mode becomes
\begin{align}\label{eqa:ss4CurrentCtrlFlt}
\left\{ \begin{array}{l}
     \dot{\gamma}_{dq}(t)= B_{c1} \tau_{dq} + B_{c 2} \begin{bmatrix} i_{l d q}(t) &v_{o d q}(t) &i_{o d q}(t)
    \end{bmatrix}^{\top}, \\
     v_{i dq}^{*}(t) = C_{c} \gamma_{dq}(t)+D_{c 1} \tau_{dq} +D_{c 2}\begin{bmatrix} i_{l d q}(t) &v_{o d q}(t) &i_{o d q}(t) \end{bmatrix}^{\top} .
     \end{array}\right.
\end{align}
Based on~\eqref{eqa:ss4VoltageCtrl},~\eqref{eqa:ss4LCL}, and~\eqref{eqa:ss4CurrentCtrlFlt}, the state-space model of the inverter-based microgrid system with ground faults can be written as
\begin{align}\label{eq:ss4faultySys}
\left\{ \begin{array}{l}
     \dot{x}(t) = A_{uh} x(t) + B_{uh1} v_{odq}^{*}(t) + B_{uh2} \tau_{dq},\\
     i_{odq}(t) = Cx(t), 
\end{array}\right.
\end{align}
where the matrices~$A_f,~B_{uh1}$, and~$B_{uh2}$ are 
\begin{align*}
     A_{uh}=\begin{bmatrix}
    \mathbf{0}_{2\times 2} & \mathbf{0}_{2\times2} & B_{v 2} \\
    \mathbf{0}_{2\times 2} & \mathbf{0}_{2\times2} & B_{c 2} \\
    \mathbf{0}_{6\times 2} & B_{l1} C_{c} &A_l + B_{l1}D_{c2}
    \end{bmatrix},
    ~B_{uh1} = \begin{bmatrix}
    B_{v1} \\ \mathbf{0}_{2\times 2} \\ \mathbf{0}_{6\times 2}
    \end{bmatrix},
    ~B_{uh2} = \begin{bmatrix}
    \mathbf{0}_{2\times 2} \\ B_{c1} \\ B_{l1}D_{c1}
    \end{bmatrix}.
\end{align*}
Note that the disturbance~$d(t)$ has no effect on the system in the fault scenario because of the short circuit.

To streamline the representation of the normal and faulty models~\eqref{eqa:ss4inverter} and~\eqref{eq:ss4faultySys}, we introduce a signal~$f(t)$ to indicate the occurrence of ground faults, i.e.,
\begin{align*}
    \left\{ \begin{array}{ll}
         f(t) = 0,  &\text{no faults}, \\
         f(t) = 1,  &\text{faults happen}.
    \end{array}\right.
\end{align*}
With $f(t)$, we can express the normal and faulty models in the following unified form
\begin{align}\label{eq:ss4unifiedmodel}
    \left\{ \begin{array}{l}
     \dot{x}(t) =  \mathcal{A}(f(t)) x(t) + \mathcal{B}_u(f(t)) u(t)+ \mathcal{B}_d(f(t)) d(t),\\
      y(t) = Cx(t),
    \end{array}\right.
\end{align}
where~$u(t) = [{v_{odq}^{*}(t)}^{\top} ~\tau^{\top}_{dq}]^{\top}$ consists of the known input signals, $y(t) = i_{odq}(t)$ is the output. The dimensions of~$x(t),~u(t),~d(t)$ and~$y(t)$ are denoted by~$n_x,~n_u,~n_d$, and~$n_y$, respectively. The system matrices are
\begin{align*}
    &\mathcal{A} (f(t)) = A_h+f(t)(A_{uh}-A_h),
   ~\mathcal{B}_u(f(t))=[B_h+f(t)(B_{uh1}-B_h) \quad f(t)B_{uh2}], \\
   ~&\mathcal{B}_d(f(t)) = [\hat{\mathcal{B}}_d(f(t)) ~\check{\mathcal{B}}_d(f(t))] = (1-f(t))[\hat{B}_d ~\check{B}_d].
\end{align*}

\begin{Rem}[Discretization]
Considering that the discrete-time samplings of data are used in the realistic framework, we discretize the continuous-time state-space model~\eqref{eq:ss4unifiedmodel} when designing the fault diagnosis scheme.
In what follows, all signals are presented in the discrete-time form.
For convenience, we use the same notation for system matrices in both the continuous and discrete representations.
\end{Rem}

\subsection{Problem statement}
The objective of this work is to detect the occurrence of ground faults in the IBM system using known signals~$u$ and~$y$.
Our proposed scheme is to design a residual generator denoted by a linear transfer function~$\F$, whose output is a \mbox{scalar-valued} signal~$r$ (called residual). 
The structure is illustrated in the diagnosis component of Fig.~\ref{fig:Architecture}.
This residual~$r$ serves as an indicator of ground faults. 
Ideally, in the absence of ground faults, $r$ should remain close to zero in the presence of disturbances. 
However, $r$ can exhibit a significant increase to facilitate detection when ground faults occur.

Additionally, in real-world application scenarios, the measurement~$\tilde{y}$, which is directly fed to the residual generator, may deviate from the output of the mathematical model~$y$ due to simplifications made during the modeling process.  
Therefore, in addition to disturbances, it is essential to ensure that~$r$ remains robust against discrepancies induced by modeling uncertainties, i.e.,~$\xi = \tilde{y} - y$.
Based on the above analysis, to obtain the desired residual behavior, two questions arise naturally.
How can we design~$\F$ to: 
\textit{
\begin{enumerate}
    \item Mitigate effects of~$d$ and~$\xi$ on~$r$ in the normal mode;
    \item Enhance fault sensitivity of $r$ in the faulty mode.
\end{enumerate} }

In this work, we provide a design method of the filter~$\F$ in the DAE framework to fulfill the two requirements. 
To this end, let us introduce the shift operator~$\mfq$, i.e.,~$\mfq x(k) = x(k+1)$, and transform the discrete-time version of the unified \mbox{state-space} model~\eqref{eq:ss4unifiedmodel} into
\begin{align}\label{eq:DAE model}
    H(\mfq,f) [X] + L(f) [Y] + E(f) [\check{d}] = 0,
\end{align}
where~$X = [x^{\top} ~\hat{d}]^{\top}$ and $Y = [y^{\top} ~u^{\top}]^{\top}$. 
The matrices~$H(\mfq,f)$ is a polynomial matrix in~$\mfq$, which is 
\begin{align*}
     &H(\mfq,f) =\mfq H_1 + H_0(f) =  \begin{bmatrix}
         -\mfq I+ \mathcal{A}(f)  &\hat{\mathcal{B}}_d(f)\\
         C                        &\mathbf{0}
     \end{bmatrix}, \\ 
     &H_1 = \begin{bmatrix}
         -I &\mathbf{0} \\ \mathbf{0} &\mathbf{0} 
     \end{bmatrix},
     \quad H_0(f) = \begin{bmatrix}
          \mathcal{A}(f) &\hat{\mathcal{B}}_d(f) \\ C &\mathbf{0}
     \end{bmatrix}.
\end{align*}
The expressions of~$L(f)$ and~$E(f)$ are
\begin{align*}
    L(f) = \begin{bmatrix}
         \mathbf{0}   &\mathcal{B}_u(f)\\
         -I  &\mathbf{0}
     \end{bmatrix},
    \quad E(f) = \begin{bmatrix}
        \check{\mathcal{B}}_d(f) \\ \mathbf{0}
    \end{bmatrix}.
\end{align*}
We define~$L_0 :=L(0)$,~$L_1 :=L(1)$, and~$E_0 := E(0)$ for simplicity of expression.

The fault detection filter~$\F$ is defined in the form of
\begin{align}\label{eq:Filter}
    \F(\mfq) = \frac{N(\mfq) L_0}{a(\mfq)} ,
\end{align}
where the numerator~$N(\mfq)$ is a polynomial row vector, i.e.,~$N(\mfq) = \sum^{d_N}_{i=0}N_i \mfq^i$, $N_i \in \R^{1 \times (n_x+n_y)}$, $d_N$ is the degree of~$N(\mfq)$. 
The denominator~$a(\mfq)$ is a polynomial with a degree larger than~$d_N$ and all roots inside the unit circle so that the derived filter is strictly proper and stable. 
For simplicity of design, we fix $d_N$ and $a(\mfq)$ and only design the coefficients of~$N(\mfq)$. 
It is worth pointing out that~$a(\mfq)$ can be chosen up to the user and specific requirements, e.g., noise sensitivity and dynamic performance, which will be our future research. 

By setting~$f=0$ and multiplying from the left-hand side of~\eqref{eq:DAE model} by~$a^{-1}(\mfq)N(\mfq)$, we obtain the residual~$r$ in the normal mode as
\begin{align}\label{eq:healthresidual}
    r &= \frac{N(\mfq) L_0}{a(\mfq)}  \begin{bmatrix}
        \tilde{y} \\ u
    \end{bmatrix}  \notag\\ 
    &= -\frac{N(\mfq) H(\mfq,0)}{a(\mfq)} [X] - \frac{N(\mfq)E_0}{a(\mfq)}[\check{d}]  + \frac{N(\mfq) L_0}{a(\mfq)}  [\bar{\xi}], 
\end{align}
where~$\bar{\xi} = [\xi^{\top} ~\mathbf{0}^{\top}]^{\top}$ as we use the practical measurement $\tilde{y}$ instead of~$y$ here.

When ground faults happen, i.e.,~$f=1$, DAE model~\eqref{eq:DAE model} becomes~$H(q,1)[X]+L_1[Y] = 0$ as~$E(1) = 0$. 
It holds that~$Y = -L^{\dag}_1 H(q,1) [X]$, where~$L^{\dag}_1$ is the left inverse of~$L_1$ and it always exist as~$L_1$ has full-column rank. 
The residual~$r$ in the faulty mode becomes
\begin{align}\label{eq:faultresidual}
    r &=\frac{N(\mfq) L_0}{a(\mfq)}  \begin{bmatrix}
        \tilde{y} \\ u
    \end{bmatrix} \notag \\
    &=  -\frac{N(\mfq) L_0 L^{\dag}_1 H(q,1)}{a(\mfq)}  [X] + \frac{N(\mfq) L_0}{a(\mfq)}  [\bar{\xi}].
\end{align}
\noindent Since all the entities in~$a^{-1}(\mfq)N(\mfq)L_0[\tilde{y}^{\top} ~u^{\top}]^{\top}$ are known, it can be used to generate the residual. 
The second line of~\eqref{eq:healthresidual} and~\eqref{eq:faultresidual} characterizes the mapping relations between the unknown signals~$x,d,\xi$ and the residual~$r$ in the normal and faulty modes, respectively, based on which we can design~$N(\mfq)$ for different diagnosis purposes.

Recall the two design requirements. To ensure a sufficiently small residual, we decouple~$X$ from~$r$ in the steady state when there is no fault, i.e.,
\begin{subequations}
\begin{align}
        N(\mfq) H(\mfq,0) \big|_{\mfq = 1} = 0, \label{eq:FullDecoupDist1}
\end{align}
Furthermore, we let the transfer function from~$X$ to $r$ remain nonzero to guarantee fault sensitivity in the faulty mode, i.e.,
\begin{align}
    N(\mfq) L_0 L^{\dag}_1 H(q,1)\big|_{\mfq = 1} \neq 0. \label{eq:FullDecoupDist2}
\end{align}

We also aim to mitigate the effects of~$\xi$ on $r$, namely, the last term in~\eqref{eq:healthresidual} and~\eqref{eq:faultresidual}. 
Inspired by the approach in~\cite{esfahani2015tractable}, we tackle this problem from a data-driven perspective by training the filter with historical data of $\xi$ to enhance its robustness.
To elaborate, we can obtain $m \in \N$ instances of output differences, i.e.,~$\{\xi_1,\dots,\xi_m\}$, by running the actual system and the mathematical model simultaneously.
For each instance~$\xi_i = [\xi_i(0),\xi_i(1)\dots,\xi_i(T)]$ with~$T\in\N$, we define its contribution to~$r$ as
\begin{align*}
    r_{\xi_i} = \frac{N(\mfq)L_0}{a(\mfq)}[\bar{\xi}_i], \quad \bar{\xi}_i = [\xi_i^{\top} ~\mathbf{0}^{\top}]^{\top}.
\end{align*}

Then, we can suppress the average effects of $\xi$ by constraining the $\mathcal{L}_2$ norm of~$r_{\xi_i}$ for all~$i \in \{1,\dots,m\}$, i.e., 
\begin{align}\label{eq:Suppress rxi}
    \frac{1}{m}\sum^m_{i=1}\|r_{\xi_i}\|^2_{\mathcal{L}_2} = \frac{1}{m}\sum^m_{i=1}\left\|\frac{N(\mfq) L_0}{a(\mfq)}  [\bar{\xi}_i] \right\|^2_{\mathcal{L}_2}  \leq  \beta_1, 
\end{align}
where~$\beta_1 \in \R_+$. 
We show later the approach to constructing~$\|r_{\xi_i}\|^2_{\mathcal{L}_2}$ with a combination of the system model and historical data. 
For the non-decoupled disturbance~$\check{d}$ in~\eqref{eq:healthresidual}, we choose the same solution as above and let
\begin{align}\label{eq:Suppress rdi}
    \frac{1}{m}\sum^m_{i=1}\|r_{\check{d}_{i}}\|^2_{\mathcal{L}_2} = \frac{1}{m}\sum^m_{i=1}\left\|\frac{N(\mfq) E_0}{a(\mfq)}  [\check{d}_{i}] \right\|^2_{\mathcal{L}_2}  \leq  \beta_2,
\end{align}
\end{subequations}
where~$\beta_2 \in \R_+$, $\check{d}_{i}$ for~$i \in \{1,\dots,m\}$ is an instance of~$\check{d}$, and~$r_{\check{d}_{i}}$ denotes the contribution of~$\check{d}_{i}$ to $r$. 

\textbf{Problem.} (Data-assisted robust fault detection filter design)
\textit{
Consider the state-space model of the IBM system~\eqref{eq:ss4unifiedmodel} with three-phase symmetrical ground faults. Given $m$ instances of output discrepancies~$\xi_i$ and non-decoupled disturbances~$\check{d}_i$, find a fault detection filter~$\F$ in the form of~\eqref{eq:Filter} that satisfies conditions~\eqref{eq:FullDecoupDist1},~\eqref{eq:FullDecoupDist2},~\eqref{eq:Suppress rxi}, and~\eqref{eq:Suppress rdi}.
}

 \section{Main Results}\label{sec:main results}
In this section, we present the design method of the fault detection filter and the determination method of the threshold.

\subsection{Filter design}
Let us start by constructing~$\|r_{\xi_i}\|^2_{\mathcal{L}_2}$ mentioned above.
Note that the response of the $j$-th element of~$\xi_i$, i.e.,~$\xi_i(j)$, can be computed by 
\begin{align*}
    \left[r_{\xi_i(j)}(0),r_{\xi_i(j)}(1),\dots,r_{\xi_i(j)}(T)\right] = N(\mfq) L_0 \bar{\xi}_i(j) \bar{\ell}_j ,
\end{align*}
where~$\bar{\ell}_j = [\overbrace{0,\dots,0}^{j},\ell(0),\ell(1),\dots,\ell(T-j)]$ and $\ell(k)$ for~$k \in \mathbb{N}$ is the value of the discrete-time unit impulse response of~$a^{-1}(\mfq)$ at the time instance $k$.
By summing up the response of~$\xi_i(j)$ for~$j \in \{0,\dots,T-d_N\}$, we obtain
\begin{align}\label{eq:rd seq}
    \left[r_{\xi_i}(0),r_{\xi_i}(1),\dots,r_{\xi_i}(T)\right] 
    &= N(\mfq)L_0 \sum^{T-d_N}_{j=0} \bar{\xi}_i(j) \bar{\ell}_j \notag\\
    =&\bar{N} \bar{L}_0
    \begin{bmatrix}
        I \\ \mfq I \\ \vdots \\ \mfq^{d_N} I
    \end{bmatrix}
    [\bar{\xi}_i(0),\dots,\bar{\xi}_i(T-d_N)] 
    \begin{bmatrix}
      \bar{\ell}_0  \\ \vdots \\ \bar{\ell}_{T-d_N}
    \end{bmatrix}, 
\end{align}
where~$\bar{N} = [N_0,N_1\dots,N_{d_N}],~\bar{L}_0 = \textup{diag}(L_0,\dots,L_0)$ according to the multiplication rule of polynomial matrices~\cite[Lemma 4.2]{esfahani2015tractable}. 
Recall that~$\mfq$ is a time shift operator, i.e.,~$\mfq \bar{\xi}_i(k) = \bar{\xi}_i(k+1)$. The equation~\eqref{eq:rd seq} becomes
\begin{align}\label{eq:rd compact form}
    \left[r_{\xi_i}(0), ~r_{\xi_i}(1),~\dots,~r_{\xi_i}(T)\right] 
    = \bar{N} \bar{L}_0 \begin{bmatrix}
        \bar{\xi}_i(0) &\dots &\bar{\xi}_i(T-d_N) \\
        \bar{\xi}_i(1) &\dots &\bar{\xi}_i(T-d_N+1) \\
        \vdots &\ddots &\vdots \\
        \bar{\xi}_i(d_N) &\dots &\bar{\xi}_i(T)
    \end{bmatrix}
     \begin{bmatrix}
      \bar{\ell}_0  \\ \vdots \\ \bar{\ell}_{T-d_N}
    \end{bmatrix} 
    = \bar{N} \bar{L}_0 \Xi_i \Gamma.
\end{align}
To ensure the existence of~$\Xi_i$, we assume that the length of data~$T$ is greatly larger than~$d_N+1$, i.e.,~$T \gg d_N+1$.
With~\eqref{eq:rd compact form}, $\|r_{\xi_i}\|^2_{\mathcal{L}_2}$ is further formulated into
\begin{align}\label{eq:rd QP}
    \|r_{\xi_i}\|^2_{\mathcal{L}_2} = \bar{N} \Phi_i \bar{N}^{\top}, \quad \Phi_i = \bar{L}_0 \Xi_i \Gamma (\bar{L}_0 \Xi_i \Gamma)^{\top}.
\end{align}
It is worth emphasizing that~$\Phi_i$ is positive semi-definite since~$\|r_{\xi_i}\|^2_{\mathcal{L}_2} = \bar{N} \Phi_i \bar{N}^{\top} \geq 0$ for all non-zero~$\bar{N}$. 
Similarly, we obtain the signature matrix for~$\check{d}_{i}$, which is,
\begin{align*}
\|r_{\check{d}_i}\|^2_{\mathcal{L}_2} = \bar{N} \Psi_i \bar{N}^{\top}, \quad \Psi_i = \bar{E}_0 \check{D}_{i}\Gamma (\bar{E}_0 \check{D}_{i}\Gamma)^{\top}.
\end{align*}
The construction of~$\bar{E}_0$ and~$\check{D}_{i}$ is similar to that of~$\bar{L}_0$ and~$\Xi_i$.

Now, we can present the design method of the ground fault detection filter for the IBM system in the following theorem.

\begin{Thm}[Filter design: QP]
\label{thm: FD QP}
Consider the unified state-space model of the IBM system~\eqref{eq:ss4unifiedmodel} and the structure of the fault detection filter in~\eqref{eq:Filter}. 
Given the degree~$d_N$, a stable~$a(\mfq)$, and $m$ instances of output discrepancies~$\xi_i$ and non-decoupled disturbances~$\check{d}_{i}$, conditions~\eqref{eq:FullDecoupDist1},~\eqref{eq:FullDecoupDist2},~\eqref{eq:Suppress rxi}, and~\eqref{eq:Suppress rdi} in Problem are satisfied by solving the following optimization problem
\begin{align} \label{eq:QP}
    \min_{\bar{N}} ~&\bar{N} (\bar{\Phi}+\bar{\Psi})\bar{N}^{\top} -\|\bar{N} \bar{L} \bar{H}(1)\bar{I}\|_{\infty}  \notag\\
    \textup{s.t.} ~&\bar{N} \bar{H}(0)\bar{I} = 0, 
\end{align}
where~$\bar{\Phi} = \frac{1}{m}\sum^{m}_{i=1} \Phi_i$,~$\bar{\Psi} = \frac{1}{m}\sum^{m}_{i=1} \Psi_i,~\bar{L} = \textup{diag}\underbrace{(L_0 L^{\dag}_1, ~\dots, ~L_0 L^{\dag}_1)}_{d_N+1},~\bar{I} = [\underbrace{I,\dots,I}_{d_N+2}]^{\top}$, and
\begin{align*}
    \bar{H}(f) = \begin{bmatrix}
        H_0(f) &H_1 &0 &\dots &0\\
        0 &H_0(f) &H_1 &0 &\vdots \\
        \vdots &0 &\ddots &\ddots &0\\
        0 &\dots &0 &H_0(f) &H_1
    \end{bmatrix}, ~f\in\{0,1\}.
\end{align*} 
 \end{Thm}
 
\begin{proof}
According to the multiplication rule of polynomial matrices, we have
\begin{align*}
    &N(\mfq) H(\mfq,0) = \bar{N} \bar{H}(0) [I, ~\mfq I, ~\dots, ~\mfq^{d_N+1} I]^{\top},\\
    &N(\mfq) L_0 L^{\dag}_1 H(q,1) = \bar{N}\bar{L}\bar{H}(1) [I, ~\mfq I, ~\dots, ~\mfq^{d_N+1} I]^{\top}. 
\end{align*}
One can see from the first equality that, by letting~$\mfq = 1$, ~$\bar{N}\bar{H}(0)\bar{I}=0$ directly implies condition~\eqref{eq:FullDecoupDist1}.
For~\eqref{eq:FullDecoupDist2}, we let coefficients of~$N(\mfq) L_0 L^{\dag}_1 H(q,1)$ be nonzero by maximizing~$\|\bar{N} \bar{L} \bar{H}(1)\bar{I}\|_{\infty}$ in the objective function, such that~\eqref{eq:FullDecoupDist2} is satisfied.
The first term in the objective function, i.e.,~$\bar{N} (\bar{\Phi} + \bar{\Psi} ) \bar{N}^{\top}$, relates to conditions~\eqref{eq:Suppress rxi} and~\eqref{eq:Suppress rdi}, which ensures that the average effects of different instances of~$\xi$ and~$\check{d}$ on $r$ are bounded. 
The derivation process of the quadratic form of~$\|r_{\xi_i}\|^2_{\mathcal{L}_2}$ and~$\|r_{\check{d}_{i}}\|^2_{\mathcal{L}_2}$ is presented in~\eqref{eq:rd seq}-\eqref{eq:rd QP}. 
This completes the proof.
\end{proof}

Note that~$\bar{N} \bar{L} \bar{H}(1)\bar{I}$ is a row vector with $(n_x+1)$ columns.
For a positive scalar~$\zeta$, $\|\bar{N} \bar{L} \bar{H}(1)\bar{I}\|_{\infty} \geq \zeta$ holds if and only if~$\bar{N} \bar{L} \bar{H}(1)\bar{I} e_i \geq \zeta$ or $\bar{N} \bar{L} \bar{H}(1)\bar{I} e_i \leq -\zeta$, where~$e_i$ is a $(n_x+1)$-dimensional column vector with only the~$i$-th element be~$1$ and the rest are zero, i.e.,~$e_i = [0, \dots,1,\dots,0]^{\top}$. 
Moreover, it is easy to check that if~$\bar{N}^*$ is a solution to~\eqref{eq:QP}, so is~$-\bar{N}^*$. 
Additionally,~$\Phi_i$ and~$\Psi_i$ are positive semi-definite.
Therefore,~\eqref{eq:QP} can be viewed as a set of $(n_x+1)$ QP problems by replacing the term~$\|\bar{N} \bar{L} \bar{H}(1)\bar{I}\|_{\infty}$ with~$\bar{N} \bar{L} \bar{H}(1)\bar{I} e_i$ (or $-\bar{N} \bar{L} \bar{H}(1)\bar{I} e_i$), and thus is convex and tractable.

\begin{Rem}[Feasibility analysis]
    It holds that~$(d_N+1)(n_x+n_y) = \textup{Rank}(\bar{H}(0)\bar{I}) + \textup{Null}(\bar{H}(0)\bar{I})$ based on Rank Plus Nullity theorem, where~$\textup{Rank}(\bar{H}(0)\bar{I})$ and~$\textup{Null}(\bar{H}(0)\bar{I})$ denote the rank and the left null space dimension of~$\bar{H}(0)\bar{I}$, respectively. 
    Thus, the equality constraint in~\eqref{eq:QP} is feasible when~$(d_N+1)(n_x+n_y) > \textup{Rank}(\bar{H}(0)\bar{I})$, i.e.,~$\textup{Null}(\bar{H}(0)\bar{I})\neq 0$.
    For~$\|\bar{N} \bar{L} \bar{H}(1)\bar{I}\|_{\infty} \neq 0$, it requires that~$\bar{L}\bar{H}(1)\bar{I}$ does not belong to the column range space of~$\bar{H}(0)\bar{I}$, i.e.,~$\textup{Rank}([\bar{H}(0)\bar{I} ~\bar{L}\bar{H}(1)\bar{I}]) > \textup{Rank}(\bar{H}(0)\bar{I})$. Otherwise, a feasible~$\bar{N}$ to~$\bar{N}\bar{H}(0)\bar{I}=0$ leads to~$\bar{N}\bar{L}\bar{H}(1)\bar{I}=0$. 
\end{Rem}

\begin{Rem}[Perfect setting]
In the perfect setting where the disturbance can be completely decoupled, one can solve the coefficients of the fault detection filter by letting $\bar{\Psi} = 0$ in the optimization problem~\eqref{eq:QP}.
\end{Rem}

We further propose an approximate analytical solution to~\eqref{eq:QP} in the following corollary.

\begin{Cor}[Approximate analytical solution]\label{cor: appro sol} 
Consider the optimization problem~\eqref{eq:QP}. There exists an approximate analytical solution given in the following form:
\begin{align}\label{eq:ana sol}
    \bar{N}^*(\delta) = \frac{\left(\bar{L} \bar{H}(1)\bar{I} e^*_i \right)^{\top}}{2\delta}  \left(\delta^{-1}(\bar{\Phi} + \bar{\Psi})+ \bar{H}(0)\bar{I}\bar{I}^{\top} \bar{H}^{\top}(0)\right)^{-1},
\end{align}

\noindent where~$e^*_i = \arg\,\max_{i\in\{1,\dots,n_x+1\}} |\bar{N}^*(\delta) \bar{L} \bar{H}(1)\bar{I} e_i|$ and $\delta \in \R_+$ is the Lagrange multiplier. The solution~$\bar{N}^*(\delta)$ provides an approximate solution to~\eqref{eq:QP} and will converge to the optimal solution as $\delta$ tends to~$\infty$.   
\end{Cor}

\begin{proof}
The Lagrange dual of~\eqref{eq:QP} can be obtained by penalizing the equality constraint
\begin{align*}
        \mathcal{L} (\bar{N},\delta) = \bar{N} (\bar{\Phi}+\bar{\Psi})\bar{N}^{\top} -\|\bar{N} \bar{L} \bar{H}(1)\bar{I}\|_{\infty} + \delta \|\bar{N} \bar{H}(0)\bar{I}\|^2_2.
\end{align*}
Since the optimization problem~\eqref{eq:QP} can be viewed as a set of QP problems by replacing the term~$\|\bar{N} \bar{L} \bar{H}(1)\bar{I}\|_{\infty}$ with~$\bar{N} \bar{L} \bar{H}(1)\bar{I} e_i$, the set of the dual functions is
\begin{align}\label{eq:dual fun}
    \mathcal{L}_i (\bar{N},\delta) = \bar{N} (\bar{\Phi}+\bar{\Psi})\bar{N}^{\top} -\bar{N} \bar{L} \bar{H}(1)\bar{I} e_i+ \delta \|\bar{N} \bar{H}(0)\bar{I}\|^2_2.
\end{align}
By taking the partial derivative of~$\mathcal{L}_i (\bar{N},\delta)$, we have
\begin{align*}
    \frac{\partial \mathcal{L}_i (\bar{N},\delta)}{\partial \bar{N} } = 2\bar{N}(\bar{\Phi}+\bar{\Psi}) +2\delta \bar{N}\bar{H}(0)\bar{I}\bar{I}^{\top}\bar{H}^{\top}(0) - (\bar{L}\bar{H}(1)\bar{I}e_i)^{\top}.
\end{align*}
Setting the partial derivative to zero leads to
\begin{align*}
    \bar{N}_i^* = \frac{1}{2} \left(\bar{L}\bar{H}(1)\bar{I}e_i \right)^{\top} \left( \bar{\Phi}+\bar{\Psi} +\delta \bar{H}(0)\bar{I} \bar{I}^{\top}\bar{H}^{\top}1(0) \right)^{-1},
\end{align*}
which is an admissible solution to the problem with the dual function~\eqref{eq:dual fun}. By choosing~$e_i$ for~$i \in \{1,\dots,n_x+1\}$ which maximizes~$|\bar{N} \bar{L} \bar{H}(1) \bar{I} e_i|$, we obtain~\eqref{eq:ana sol}. This completes the proof.
\end{proof}

With the analytical solution, one can update the coefficients of the filter online with new data without the need to re-solve~\eqref{eq:QP}, which is an improvement over~\cite{esfahani2015tractable}.

\begin{Rem}[Average objective function]
We consider the average effects of all $\xi_i$ and~$\check{d}_i$ on the residual as the objective function in~\eqref{eq:QP}.
An alternative way is to consider the worst-case scenario, i.e.,~$\max_{i\in\{1,\dots,m\}} \bar{N}(\Phi_i+\Psi_i)\bar{N}^{\top}$. 
The average objective function is, however, of interest if one requires to train the filter with a large number of patterns. This is due to the fact that the computational complexity of the derived QP problem is independent of the number of instances with the average objective function.  
\end{Rem}

\begin{Rem}[Approximate analytical solution with~$\delta$]
The Lagrange multiplier~$\delta$ is introduced in~\eqref{eq:QP} to penalize the equality constraint~$\bar{N} \bar{H}(0) \bar{I} = 0$, and in the ideal case, $\delta$ tends to infinity as stated in Corollary~\ref{cor: appro sol}. 
However, for a bounded~$\delta$, the equality constraint cannot be strictly satisfied, which is why we refer to the solution~\eqref{eq:ana sol} as an approximate analytical solution. Additionally, to ensure that~$\bar{N} \bar{H}(0) \bar{I}$ is sufficiently close to zero, $\delta$ should be large enough while remaining numerically bounded for practical considerations.
\end{Rem}

\begin{Rem}[Control saturation]
    In our problem, we consider a small-signal model of the IBM system with bounded load changes, thus controller saturation is rare and not a major problem we aimed to deal with. Nonetheless, in the case of saturation, we can address this issue by modeling the IBM system with saturation and incorporating the disturbance suppression when saturation happens in design conditions. 
\end{Rem}

To detect the fault, we introduce the power of the residual~$r(k)$ as the evaluation function, i.e.,~$J(r) = r(k)^2$ for~$k \in \N$. Let~$J_{th}$ be the detection threshold. Then, we can consider the following detection logic:
\begin{align*}
    \left\{ \begin{array}{lll}
        J(r) \leq J_{th} & \Rightarrow & \textup{no faults},\\
        J(r) > J_{th}    & \Rightarrow & \textup{faults}.
    \end{array}\right.
\end{align*}
We show the computation method of the threshold and the false alarm rate in the following proposition.

\begin{Prop}[Probabilistic false alarm certificate]\label{prop:prob perform}
Assume that the patterns of the non-decoupled disturbance~$\check{d}$ and the modeling uncertainty~$\xi$ follow different i.i.d. distributions. Consider the system~\eqref{eq:ss4unifiedmodel}, the filter~$\F(\mfq)$ obtained by solving~\eqref{eq:QP} with the corresponding optimal solution~$\bar{N}^*$, and the evaluation function~$J(r)=r(k)^2$. 
Given a scalar~$\lambda \geq 1$, if we set the threshold~$J_{th}$ as
\begin{align}\label{eq:Threshold}
    J_{th}  = \frac{\lambda}{T} \bar{N}^*(\bar{\Phi}+\bar{\Psi})\bar{N}^{*\top},
\end{align}
the false alarm rate in the steady state satisfies 
\begin{align}
    &\lim\limits_{k \rightarrow \infty} \mPr\{J(r(k)) > J_{th} | f = 0\} 
    \leq \frac{1}{\lambda}.
\end{align}
\end{Prop}

\begin{proof}
Since both~$\check{d}_i$ and~$\xi_i$ follow i.i.d. distributions, the residual in the normal mode as shown in~\eqref{eq:healthresidual} can be viewed as a random variable. 
It is proven in~\cite[Theorem 4.11]{esfahani2015tractable} that the empirical average error
\begin{align*}
    \varepsilon_m = \frac{1}{m}\sum^{m}_{i=1} \|r_{\xi_i} + r_{\check{d}_i}\|^2_{\mathcal{L}_2}  -  \mE\left[\|r\|^2_{\mathcal{L}_2}\right],
\end{align*}
satisfies the strong law of large numbers, i.e.,~$\lim_{m \rightarrow \infty} \varepsilon_m=0$  almost surely. Therefore, in the steady state, it holds that
\begin{align*}
     J_{th} = &\lim_{T,m \rightarrow \infty} \frac{\lambda}{T} \left(\frac{1}{m}\sum^m_{i=1} \|r_{\xi_i}\|^2_{\mathcal{L}_2} +  \frac{1}{m}\sum^m_{i=1} \|r_{\check{d}_i}\|^2_{\mathcal{L}_2} \right)\\
    \geq &\lim_{T,m \rightarrow \infty} \frac{\lambda}{T} \left(\frac{1}{m}\sum^m_{i=1} \|r_{\xi_i}+r_{\check{d}_i}\|^2_{\mathcal{L}_2} \right)\\
    = &\lim_{T \rightarrow \infty}\frac{\lambda}{T} \mE\left[\|r\|^2_{\mathcal{L}_2}\right]  
    = \lambda \lim_{k \rightarrow \infty} \mE[r(k)^2].
\end{align*}
According to Markov inequality, the false alarm rate in the steady state satisfies
\begin{align*}
  \lim\limits_{k \rightarrow \infty} \mPr\{ r(k)^2 >J_{th} | f = 0\} 
  <  \lim\limits_{k \rightarrow \infty} \mPr\{ r(k)^2 >\lambda \mE[r(k)^2] | f = 0\} \leq \frac{1}{\lambda}.
\end{align*}
This completes the proof.
\end{proof}

We further derive the circumstances in which ground faults can be detected by comparing the steady-state value of $r^2$ with~$J_{th}$. 
However, given that the faulty model is unobservable, it is essential to first identify its observable subsystem, denoted as~$(A_{uh,o},[B_{uh1,o}, B_{uh2,o}], C_{uh,o})$, through Kalman decomposition.
Define the transfer function from~$[{v^*_{odq}}^{\top} ~\tau_{dq}]^{\top}$ to $r$ as $\mathcal{T}_{ur}(\mfq) = C_{uh,o}(\mathfrak{q}I-A_{uh,o})^{-1}[B_{uh1,o}, B_{uh2,o}]$.
Then, the ground faults can be detected if
    \begin{align*}        \left(\frac{N(\mathfrak{q})L_0}{a(\mathfrak{q})} \begin{bmatrix}
           \mathcal{T}_{ur}(\mfq) \\ I 
        \end{bmatrix} \begin{bmatrix}
            v^*_{odq} \\ \tau_{dq}
        \end{bmatrix} {\Bigg|}_{\mathfrak{q}=1}  \right) ^2 > J_{th}.
    \end{align*}
    
\section{Numerical results in RTDS} \label{sec:simulation}
In this section, we validate the performance of the fault detection filters through numerical simulations. 
Here, we model the IBM system depicted in Fig.~\ref{fig:Architecture} through the mathematical model (MM)~\eqref{eq:ss4unifiedmodel}, Simulink, and RTDS, respectively.
Note that electrical components integrated into Simulink and RTDS allow for a more realistic simulation of practical scenarios compared to the simplified MM.
In addition, RTDS is widely recognized in the industry for its ability to simulate power systems in real time~\cite{podmore2010role}.
Subsequently, we collect the discrepancy data between the output of MM and those from Simulink and RTDS. Step signals are employed to characterize the unknown load changes. 
Following this, we design the filter by solving the optimization problem~\eqref{eq:QP}. 
Finally, we apply the derived filter to models constructed by Simulink and RTDS to evaluate diagnosis performance in the presence of partially decoupled disturbances and modeling uncertainties.

The parameters of the system are presented in Table~\ref{tab:para} and Table~\ref{tab:init}, sourced from~\cite{pogaku2007modeling} with some modifications. The sampling period is~$0.1$~ms and the simulation time is $6$~s. The reference frequency and reference voltage are~$50$~Hz and $v^*_{odq} = [381,0]^{\top}$, respectively. 
The FCL parameter is~$\tau_{dq} = [35,0.7]^{\top}$.
To design the fault detection filter following the structure of~\eqref{eq:Filter}, we set the degree of~$N(\mfq)$ as~$d_N = 10$ and choose a stable denominator~$a(\mfq)$ with a degree larger than~$d_N$.
We further collect $100$ output discrepancy and disturbance instances with $T = 200$ to construct~$\bar{\Phi}$ and $\bar{\Psi}$, respectively. 
The disturbance is a step signal, i.e., $d = [-15 ~0.1]^{\top}$ for~$k>15000$.
The ground fault occurs at~$k = 40000$.
Given the settings specified above, we utilize the YALMIP toolbox~\cite{Lofberg2004} to address the optimization problem~\eqref{eq:QP} and obtain the fault detection filter.

\begin{minipage}{\textwidth}
\centering
 \begin{minipage}[h]{0.49\textwidth}
  \centering   \makeatletter\def\@captype{table}\makeatother\caption{Microgrid parameters.}\label{tab:para}
        \begin{tabular}{lr |lr}
        \hline Parameter & Value & Parameter & Value \\
        \hline$\omega$ & $314.1 \mathrm{~Hz}$ & $R_{LOAD}$ & $12 \ \Omega$\\
        $L_{f}$ & $3.5 \ \mathrm{mH}$ & $K_P^c$ & 0.3\\
        $R_{f}$ & $0.01 \ \Omega$ & $K_I^c$ & $20$\\
        $C_{f}$ & $21.9 \mu \mathrm{~F}$ & $K_P^v$ & 2 \\
        $L_{c}$ & $1.3 \mathrm{~mH}$ &   $K_I^v$ & 14 \\
        $R_{c}$ & $0.02 \ \Omega$ & $F$ & $0.75$ \\
        \hline
        \end{tabular}
  \end{minipage}
  \begin{minipage}[h]{0.49\textwidth}
   \centering       \makeatletter\def\@captype{table}\makeatother\caption{Initial conditions.}\label{tab:init}
        \begin{tabular}{lr |lr}
        \hline Parameter & Value & Parameter & Value \\
        \hline
        $v_{od}$ & $380.8$  & $i_{ld}$ & $11.4$ \\
        $v_{oq}$ & $0 $        & $i_{lq}$ & $-5.5 \times 10^3$\\
        $i_{od}$ & $11.4$ & $v_{bd}$ & 379.5 \\
        $i_{oq}$ & $0.4$   &   $v_{bq}$ & -6 \\
        $\phi_d$ & $0.13$  &   $\gamma_d$ & $0.0115$\\
        $\phi_q$ & $0$     &   $\gamma_q$ & $0$\\
        \hline
        \end{tabular}
   \end{minipage}
\end{minipage}

Fig.~\ref{fig: Output} illustrates the per-unit (p.u.) values of output currents generated by the three different models, which exhibit close resemblance. 
This implies that the simplified MM aligns well with the more intricate systems constructed using Simulink and RTDS. From the small figures, one can see minor discrepancies between these outputs. 
Furthermore, there is a step change in the output currents after an unknown load change at~$k=15000$, which is similar to the effect caused by the ground fault at $k=40000$. Therefore, discerning the occurrence of a ground fault solely from the output current proves challenging.

Fig.~\ref{fig: Residual} displays the diagnosis results characterized through the values of~$r^2(k)$. 
In the left figure, we show the diagnostic performance of the filter designed to withstand output discrepancies between MM and the Simulink model. 
One can see that~$r^2(k)$ remains below the threshold in the presence of partially decoupled disturbances and modeling uncertainties until the occurrence of the fault at~$k=40000$. 
After the fault occurs, $r^2(k)$ immediately exceeds the threshold and remains higher than the threshold. 
This indicates that the fault is successfully detected with modeling uncertainties and is distinguishable from disturbances.
The right-hand side figure depicts the diagnosis outcome of the filter designed for the RTDS model. 
One can see that this filter effectively suppresses the partially decoupled disturbances and modeling uncertainties and successfully detects the ground fault as well. 
Note that due to the spike induced by the load change in the output current, $r^2(k)$ surpasses the threshold and causes a false alarm. 
Nonetheless, the signal rapidly diminishes below the threshold, distinguishing it from scenarios where ground faults occur. 
To address this issue, we can extend the evaluation time.
\begin{figure}[h]
    \centering
    \includegraphics[width=0.49\columnwidth]{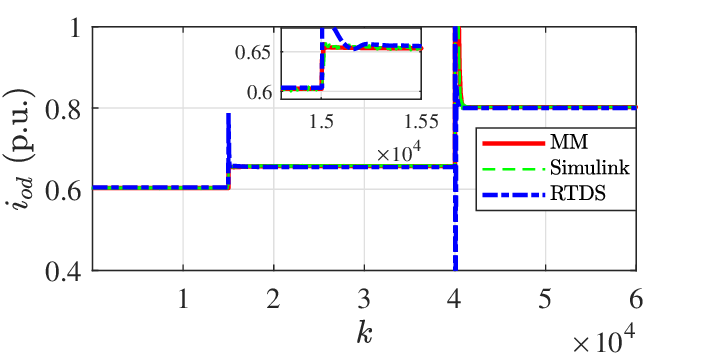}
    \hfill
    \includegraphics[width=0.49\columnwidth]{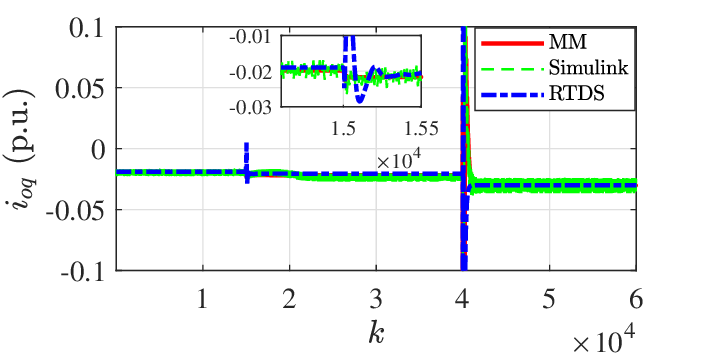}
    \caption{\small Output currents generated by different models.}\label{fig: Output}
\end{figure}

\begin{figure}[h]
    \centering
    \includegraphics[width=0.49\columnwidth]{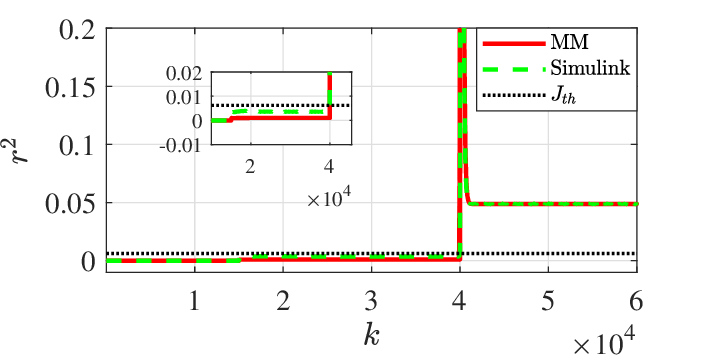}
    \hfill
    \includegraphics[width=0.49\columnwidth]{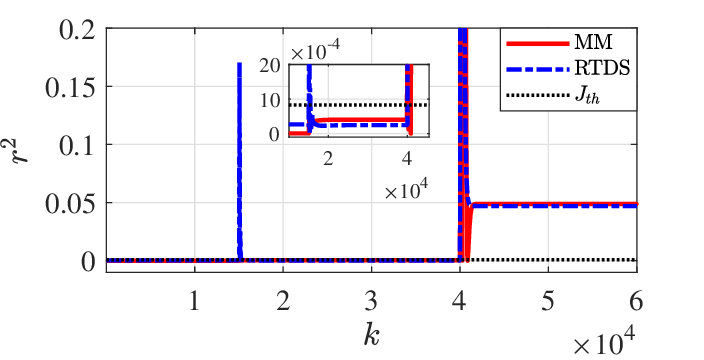} 
    \caption{\small Diagnosis results with different models.}\label{fig: Residual}
\end{figure}


\section{Conclusions} \label{sec:conslusion}
In this paper, we propose a diagnosis scheme for the detection of ground faults in IBM systems with partially decoupled disturbances and modeling uncertainties.
To this end, we reformulate the filter design problem into a QP problem, which enables us to efficiently optimize the filter parameters and meet the desired performance criteria, including decoupling partial disturbances, suppressing effects of non-decoupled disturbances and modeling uncertainties, and ensuring fault sensitivity.
Simulation results on an IBM system constructed using MM, Simulink, and RTDS, respectively, show the effectiveness of the proposed approach.
In future work, we first will consider designing the denominator of the filter for better dynamic performance. The second direction will be focused on extending the proposed approaches to more complex and realistic settings, such as considering the presence of multiple converters.


\appendix
\section{Simulation results for the perfect setting}
In the appendix, we additionally present the simulation results for the perfect scenario where disturbances can be completely decoupled, i.e.,$n_d = 1$, and where uncertainties are absent. 
In particular, we set the matrix~$B_d = \begin{bmatrix}
         \mathbf{0}_{1 \times 8} & [1 ~1] 
         \end{bmatrix}^{\top}$. 
The disturbance is zero for~$k \leq 1000$, and subsequently follows a signal given by~$d(k) = \alpha_0 +\sum^{\eta}_{i=1} \alpha_{i} \sin (\omega_{i} k + \psi_{i})$ for~$k>1000$.
Specifically, the constant~$\alpha_0 \in \R$ represents an abrupt change, while the sinusoidal terms capture the short-term load fluctuations with amplitudes~$\alpha_i$, angular frequencies~$\omega_i \in \R_+$, and phases~$\psi_{i} \in \R$~\cite{esfahani2015tractable}. 
It is worth emphasizing that we deliberately select the parameters of~$B_d$ and magnitude of~$d$ to make the output currents similar in the faulty mode and under the effect of the disturbance, which increases the difficulty of fault detection. We design the fault detection filter using the method proposed in Theorem~\ref{thm: FD QP}. 
As non-decoupled disturbances and modeling uncertainties are disregarded, we assign $\bar{\Phi}=\bar{\Psi} = 0$ within optimization problem~\eqref{eq:QP} for the determination of filter coefficients.
The simulation results are presented in Figure~\ref{fig:Decoupled simulation}.
Note that, different from the p.u. used in the above simulation, the unit of current utilized in the appendix is Amperes.

\begin{figure}[h]
    \centering
    \subfigure[Small load fluctuations.]{\includegraphics[width=0.45\columnwidth]{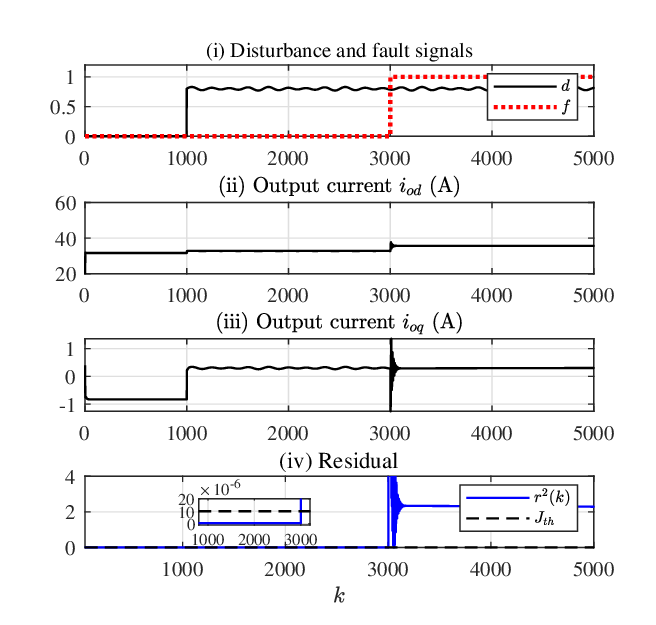}}
    \subfigure[Large load fluctuations.]{\includegraphics[width=0.45\columnwidth]{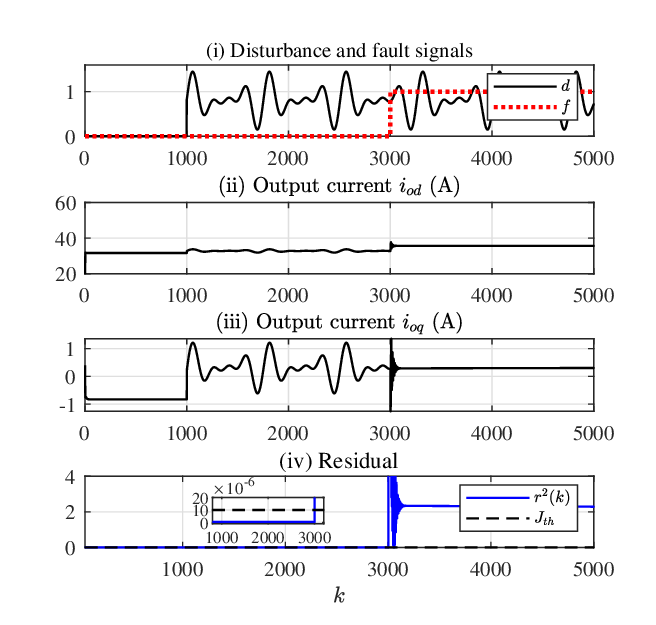}}  
    \caption{\small Diagnosis results for the perfect setting.}\label{fig:Decoupled simulation}
\end{figure}

Figure~\ref{fig:Decoupled simulation}a presents the diagnosis results when a decoupled disturbance has small load fluctuations, i.e.,~$d(k) = 0.8+0.02\sin(k/30)+0.01\sin(k/40)+0.01\sin(k/60)$.
As shown in Figure~\ref{fig:Decoupled simulation}a (i), the disturbance~$d$ and the ground fault~$f$ occur at $k = 1001$ and~$k=3001$, respectively. 
However, $d$ and $f$ have similar effects on the output currents~$i_{od}$ and $i_{oq}$ from Figure~\ref{fig:Decoupled simulation}a (ii) and (iii), which only exhibit minor variations. 
This makes it challenging to detect the occurrence of the ground fault and distinguish it from the disturbance only through the output currents.
In contrast, Figure~\ref{fig:Decoupled simulation}a (iv) illustrates that the residual is insensitive to the disturbance and stays below the threshold until the fault happens.
The power value of the residual $r^2(k)$ exceeds the threshold at $k=3002$, resulting in the detection of the fault within 0.1 ms. 
We further consider a decoupled disturbance with larger load fluctuations, i.e.,~$d(k) = 0.8+0.2\sin(k/30)+0.3\sin(k/40)+0.2\sin(k/60)$. 
Figure~\ref{fig:Decoupled simulation}b displays the diagnosis results and the analysis process is analogous to the previous one.

\section{Comparison with other approaches}
    We employ the wavelet transform approach to analyze the output current signal, which is commonly used for fault detection in the electrical engineering field~\cite{rivas2020faults}.  
    The result is presented in Fig.~\ref{fig: wavelet}. One can observe that the magnitude of the transformed signal increases in a very short time after the fault occurs at $4 \times 10^4$. Then, it reverts to the previous size, making it difficult to distinguish the ground fault from the fault-free case based solely on the wavelet transform result.

    In addition, considering that $H_\infty$ method is widely used in the literature to deal with disturbances and uncertainties, we also use it to compare the diagnosis performance with our proposed data-assisted approach.
    The simulation result is shown in Fig.~\ref{fig: hinf}.
    We can see that though the fault can be detected, the value of $r^2$ is quite small after the fault happens compared to the results obtained using the data-assisted approach, i.e.,~$4.13 \times 10^{-5}$, as shown in Fig.~\ref{fig: Residual}.
    This is because the $H_\infty$ approach is designed to handle all signals, thereby effectively suppressing disturbances and uncertainties but also impacting fault sensitivity greatly. 
    In contrast, our data-assisted approach is more focused and tailored, thus can offer better fault sensitivity. 

    \begin{figure}[h]
        \centering
        \begin{minipage}[c]{0.48\textwidth}
             \includegraphics[width = \textwidth]{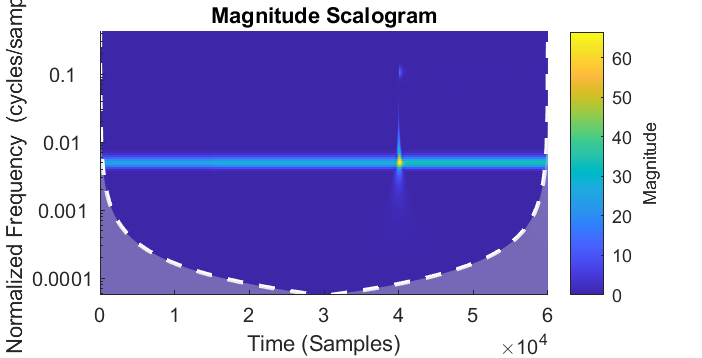}
             \caption{\small Diagnosis results using wavelet transform analysis.}
        \label{fig: wavelet}
        \end{minipage}
        \begin{minipage}[c]{0.48\textwidth}
             \includegraphics[width = \textwidth]{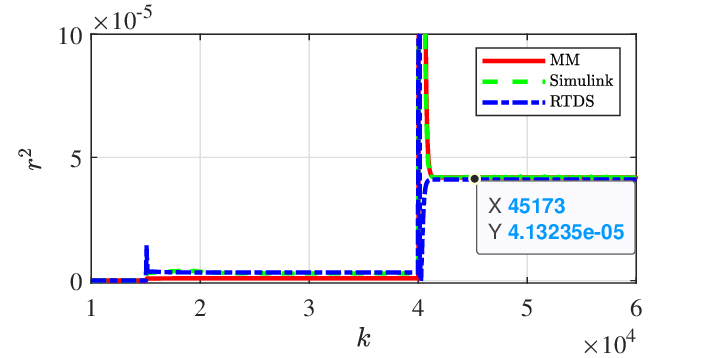}
        \caption{\small Diagnosis results using the $H_\infty$ approach.}
        \label{fig: hinf}
        \end{minipage}
    \end{figure}

\bibliographystyle{elsarticle-num}
\bibliography{ref}

\begin{thebibliography}{10}
\expandafter\ifx\csname url\endcsname\relax
  \def\url#1{\texttt{#1}}\fi
\expandafter\ifx\csname urlprefix\endcsname\relax\def\urlprefix{URL }\fi
\expandafter\ifx\csname href\endcsname\relax
  \def\href#1#2{#2} \def\path#1{#1}\fi

\bibitem{altaf2022microgrid}
M.~W. Altaf, M.~T. Arif, S.~N. Islam, M.~E. Haque, Microgrid protection
  challenges and mitigation approaches-{A} comprehensive review, IEEE Access
  (2022).

\bibitem{zamani2012communication}
M.~A. Zamani, A.~Yazdani, T.~S. Sidhu, A communication-assisted protection
  strategy for inverter-based medium-voltage microgrids, IEEE Transactions on
  Smart Grid 3~(4) (2012) 2088--2099.

\bibitem{lai2015comprehensive}
K.~Lai, M.~S. Illindala, M.~A. Haj-ahmed, Comprehensive protection strategy for
  an islanded microgrid using intelligent relays, in: 2015 IEEE Industry
  Applications Society Annual Meeting, IEEE, 2015, pp. 1--11.

\bibitem{karimi2019protection}
H.~Karimi, G.~Shahgholian, B.~Fani, I.~Sadeghkhani, M.~Moazzami, A protection
  strategy for inverter-interfaced islanded microgrids with looped
  configuration, Electrical Engineering 101~(3) (2019) 1059--1073.

\bibitem{loix2009protection}
T.~Loix, T.~Wijnhoven, G.~Deconinck, Protection of microgrids with a high
  penetration of inverter-coupled energy sources, in: 2009 CIGRE/IEEE PES Joint
  Symposium Integration of Wide-Scale Renewable Resources Into the Power
  Delivery System, IEEE, 2009, pp. 1--6.

\bibitem{casagrande2013differential}
E.~Casagrande, W.~L. Woon, H.~H. Zeineldin, D.~Svetinovic, A differential
  sequence component protection scheme for microgrids with inverter-based
  distributed generators, IEEE Transactions on Smart Grid 5~(1) (2013) 29--37.

\bibitem{samantaray2012differential}
S.~Samantaray, G.~Joos, I.~Kamwa, Differential energy based microgrid
  protection against fault conditions, in: 2012 IEEE PES Innovative Smart Grid
  Technologies (ISGT), IEEE, 2012, pp. 1--7.

\bibitem{sortomme2009microgrid}
E.~Sortomme, S.~Venkata, J.~Mitra, Microgrid protection using
  communication-assisted digital relays, IEEE Transactions on Power Delivery
  25~(4) (2009) 2789--2796.

\bibitem{liu2020fast}
L.~Liu, Z.~Liu, M.~Popov, P.~Palensky, M.~A. van~der Meijden, A fast protection
  of multi-terminal {HVDC} system based on transient signal detection, IEEE
  Transactions on Power Delivery 36~(1) (2020) 43--51.

\bibitem{karimi2008negative}
H.~Karimi, A.~Yazdani, R.~Iravani, Negative-sequence current injection for fast
  islanding detection of a distributed resource unit, IEEE Transactions on
  Power Electronics 23~(1) (2008) 298--307.

\bibitem{pirani2022optimal}
M.~Pirani, M.~Hosseinzadeh, J.~A. Taylor, B.~Sinopoli, Optimal active fault
  detection in inverter-based grids, IEEE Transactions on Control Systems
  Technology (2022).

\bibitem{gao2015survey}
Z.~Gao, C.~Cecati, S.~X. Ding, A survey of fault diagnosis and fault-tolerant
  techniques—part {I}: Fault diagnosis with model-based and signal-based
  approaches, IEEE Transactions on Industrial Electronics 62~(6) (2015)
  3757--3767.

\bibitem{gao2015unknown}
Z.~Gao, X.~Liu, M.~Z. Chen, Unknown input observer-based robust fault
  estimation for systems corrupted by partially decoupled disturbances, IEEE
  Transactions on Industrial Electronics 63~(4) (2015) 2537--2547.

\bibitem{nyberg2006residual}
M.~Nyberg, E.~Frisk, Residual generation for fault diagnosis of systems
  described by linear differential-algebraic equations, IEEE Transactions on
  Automatic Control 51~(12) (2006) 1995--2000.

\bibitem{esfahani2015tractable}
P.~Mohajerin~Esfahani, J.~Lygeros, A tractable fault detection and isolation
  approach for nonlinear systems with probabilistic performance, IEEE
  Transactions on Automatic Control 61~(3) (2015) 633--647.

\bibitem{pan2021dynamic}
K.~Pan, P.~Palensky, P.~Mohajerin~Esfahani, Dynamic anomaly detection with
  high-fidelity simulators: A convex optimization approach, IEEE Transactions
  on Smart Grid 13~(2) (2021) 1500--1515.

\bibitem{dong2023multimode}
J.~Dong, A.~S. Kolarijani, P.~Mohajerin~Esfahani, Multimode diagnosis for
  switched affine systems with noisy measurement, Automatica 151 (2023) 110898.

\bibitem{van2022multiple}
C.~Van~der Ploeg, M.~Alirezaei, N.~Van De~Wouw, P.~Mohajerin~Esfahani, Multiple
  faults estimation in dynamical systems: Tractable design and performance
  bounds, IEEE Transactions on Automatic Control 67~(9) (2022) 4916--4923.

\bibitem{van2021real}
C.~van~der Ploeg, E.~Silvas, N.~van~de Wouw, P.~Mohajerin~Esfahani, Real-time
  fault estimation for a class of discrete-time linear parameter-varying
  systems, IEEE Control Systems Letters 6 (2021) 1988--1993.

\bibitem{pogaku2007modeling}
N.~Pogaku, M.~Prodanovic, T.~C. Green, Modeling, analysis and testing of
  autonomous operation of an inverter-based microgrid, IEEE Transactions on
  Power Electronics 22~(2) (2007) 613--625.

\bibitem{leitner2017small}
S.~Leitner, M.~Yazdanian, A.~Mehrizi-Sani, A.~Muetze, Small-signal stability
  analysis of an inverter-based microgrid with internal model-based
  controllers, IEEE Transactions on Smart Grid 9~(5) (2017) 5393--5402.

\bibitem{park1929two}
R.~H. Park, Two-reaction theory of synchronous machines generalized method of
  analysis-part {I}, Transactions of the American Institute of Electrical
  Engineers 48~(3) (1929) 716--727.

\bibitem{kahrobaeian2012interactive}
A.~Kahrobaeian, Y.~A.-R.~I. Mohamed, Interactive distributed generation
  interface for flexible micro-grid operation in smart distribution systems,
  IEEE Transactions on Sustainable Energy 3~(2) (2012) 295--305.

\bibitem{singh2017distributed}
V.~P. Singh, N.~Kishor, P.~Samuel, Distributed multi-agent system-based load
  frequency control for multi-area power system in smart grid, IEEE
  Transactions on Industrial Electronics 64~(6) (2017) 5151--5160.

\bibitem{xu2017robust}
Y.~Xu, Robust finite-time control for autonomous operation of an inverter-based
  microgrid, IEEE Transactions on Industrial Informatics 13~(5) (2017)
  2717--2725.

\bibitem{ding2008model}
S.~X. Ding, Model-based fault diagnosis techniques: design schemes, algorithms,
  and tools, Springer Science \& Business Media, 2008.

\bibitem{podmore2010role}
R.~Podmore, M.~R. Robinson, The role of simulators for smart grid development,
  IEEE Transactions on Smart Grid 1~(2) (2010) 205--212.

\bibitem{Lofberg2004}
J.~L{\"{o}}fberg, {YALMIP}: A toolbox for modeling and optimization in
  {MATLAB}, in: In Proceedings of the CACSD Conference, Taipei, Taiwan, 2004.

\bibitem{rivas2020faults}
A.~E.~L. Rivas, T.~Abrao, Faults in smart grid systems: {M}onitoring, detection
  and classification, Electric Power Systems Research 189 (2020) 106602.

\end{thebibliography}
\end{document}